\documentclass[a4paper,noarxiv,twocolumn,accepted=2022-09-01]{quantumarticle}
\usepackage[utf8]{inputenc}

\usepackage{amsmath, amsfonts, amsthm}
\usepackage[linesnumbered,ruled,vlined,boxed]{algorithm2e}
\usepackage{url}
\usepackage[breaklinks]{hyperref}

\usepackage{geometry}
\usepackage{color}
\usepackage{booktabs}

\usepackage{graphicx}
\usepackage{rotating,tabularx}
\newcommand{\algref}[1]{{\rm Algorithm~\ref{alg:#1}}}

\newtheorem{theorem}{Theorem}[section]
\newtheorem{corollary}{Corollary}[section]

\newtheorem{condition}{Condition}[section]
\newtheorem{proposition}{Proposition}[section]

\newcommand{\thetab}{\mathbf{\theta}}
\newcommand{\R}{\mathbb{R}}
\newcommand{\sba}{\mathbf{s}}

\DeclareSymbolFont{bbold}{U}{bbold}{m}{n}
\DeclareSymbolFontAlphabet{\mathbbold}{bbold}

\newcommand{\lineref}[1]{Line~\ref{line:#1}}

\definecolor{mocolor}{rgb}{0,0.5,0.5}
\definecolor{mmcolor}{rgb}{0.5,0,0.5}
\newcommand{\change}[1]{{#1}} 
\graphicspath{{./images/}}
\title{Latency considerations for stochastic optimizers in variational quantum algorithms}
\author{Matt Menickelly}
\email{mmenickelly@anl.gov}
\affiliation{Mathematics and Computer Science Division, Argonne National Laboratory, 9700 S. Cass Ave., Lemont, IL 60439}

\author{Yunsoo Ha}
\email{yha3@ncsu.edu}
\affiliation{Edward P. Fitts Department of Industrial and Systems Engineering, North Carolina State University, 915 Partners Way, Raleigh, NC 27601}
\author{Matthew Otten}
\email{mjotten@hrl.com}
\affiliation{HRL Laboratories, LLC, 3011 Malibu Canyon Road, Malibu, CA 90265}
\begin{document}
\maketitle

\begin{abstract}
Variational quantum algorithms, which have risen to prominence in the noisy intermediate-scale quantum setting,
 require the implementation of a stochastic optimizer on classical hardware. 
To date, most research has employed algorithms based on the stochastic gradient iteration as the stochastic classical optimizer. 
In this work we propose instead using stochastic optimization algorithms that yield stochastic processes emulating the dynamics of classical deterministic algorithms.
This approach results in methods with theoretically superior worst-case iteration complexities, at the expense of greater per-iteration sample (shot) complexities.
We investigate this trade-off both theoretically and empirically and conclude that preferences for a choice of stochastic optimizer should explicitly depend on a function of both latency and shot execution times. 
\end{abstract}

\section{Introduction}

Quantum computers have the potential to perform many important calculations faster than their classical
counterparts do. Disparate domains such as quantum chemistry~\cite{lanyon2010towards}; nuclear~\cite{cloet2019opportunities}, 
condensed matter~\cite{smith2019simulating}, and
high energy~\cite{nachman2021quantum} physics; machine learning and data science~\cite{biamonte2017quantum};
and  finance~\cite{orus2019quantum} are 
all theorized to benefit from quantum algorithms in various ways. In the near-term, the so-called 
noisy intermediate-scale quantum (NISQ)~\cite{preskill2018quantum} era, these theoretical benefits are hard to realize,
because the canonical quantum algorithms that are used in many of the domains, such as 
quantum phase estimation~\cite{dorner2009optimal}, require gate depths that are  expected to be realized only with
fault-tolerant, error-corrected quantum computers~\cite{preskill1998fault}. 
Variational quantum algorithms (VQAs) attempt to
lower gate depth requirements by replacing some of the  requirements with the need
to perform optimization using classical computers~\cite{cerezo2021variational}. 
Such algorithms have shown success on
NISQ hardware in eigenvalue estimation~\cite{o2016scalable}, 
dynamical evolution~\cite{yuan2019theory,otten2019noise,kandala2017hardware}, machine 
learning~\cite{mitarai2018quantum,otten2020quantum}, and many other problems~\cite{cerezo2021variational}. 
The variation quantum eigensolver (VQE)
is perhaps the best-known example, with many demonstrations performed on a variety
of physical and chemical systems~\cite{kandala2017hardware,parrish2019quantum}. 
One of the key bottlenecks in VQAs is the optimization
step; the functions being optimized, as well as their higher-order derivative information, need to be estimated via a 
limited number of samples, necessitating the use of stochastic optimization algorithms. 
A simple way of quantifying the total cost would be through the number of shots: for each
function evaluation, the quantum computer needs to be queried many times to get an estimate
of the function to be used in a classical optimization routine. These sampling queries
are in addition to the need to run many different sets of parameters along 
the optimization trajectory, which introduces a ``circuit switching'' latency. 
 Given that many of today's quantum computers are available
in a cloud setting, there can be additional overhead due to network latency~\cite{sung2020using}.
These latencies can also vary among different architectures; for example, a
superconducting transmon quantum processor has measurement times on the order of a few
microseconds~\cite{gambetta2007protocols}, while a trapped ion system can take hundreds of microseconds~\cite{clark2021engineering,bruzewicz2019trapped}. This is, of course,
in addition to gate times and reset times, all of which increase the 
time to take a single sample.
Designing optimization algorithms that take these various times into account is important
for minimizing the total amount of potentially expensive quantum computer time used. 

A few stochastic gradient methods have been developed specifically for the quantum regime
that attempt to minimize the number of shots~\cite{kubler2020adaptive}, as well as those that attempt
to minimize total wall time~\cite{sung2020using}. 
However, stochastic gradient methods,
including the popular Adam optimizer~\cite{kingma2014adam}, are typically deployed in applications where accuracy
does not matter  much. In quantum chemistry, where VQE is used extensively, there 
is a standard level of accuracy, known as ``chemical accuracy,'' that is desirable 
because it is the minimum accuracy required to allow for chemically relevant predictions~\cite{helgaker2014molecular}.
Many stochastic gradient methods also require careful hyperparameter tuning~\cite{Schaul14,asi2019importance},
particularly in the selection of a step-size parameter. 
In VQAs, each sample (or observation) can be 
 expensive because quantum computer time is highly limited. Doing extensive hyperparameter sweeps
is, therefore, untenable. 

\subsection{Our Contributions}
Given these circumstances, we suggest a new optimization algorithm, 
SHOt-Adaptive Line Search (SHOALS).
In each iteration, SHOALS computes a stochastic estimate of the gradient of the cost function, computes a trial step, and evaluates the cost function at both the current set of parameters and the trial step. If sufficient decrease within a confidence interval is detected, the step is accepted, and the trial step becomes the current set of parameters for the next iteration. 
In every iteration, SHOALS updates certain accuracy parameters that, when coupled with conservative concentration inequalities, determine sample sizes for computing function and gradient values. 

The convergence theory behind SHOALS is derived from ALOE ~\cite{jinscheinbergxie}, which in turn is an extension of previous work in ~\cite{blanchet2019convergence,paquette2020stochastic,berahas2021global}.
These works in adaptive optimization consider the quantities involved in deterministic methods of optimization (in particular, trust region methods and line search methods) and replace these quantities with random variables to yield an analyzable stochastic process. 
As a result, one derives stochastic variants of deterministic methods that yield, up to constants, the worst-case iteration complexity guarantees of the deterministic methods with high probability;
SHOALS achieves this by mimicking the dynamics of gradient descent with a line search. 
These results are significant because the iteration complexities of deterministic methods are of a different order than those of stochastic methods. 
For instance, in the worst case, given Lipschitz continuous $f$, the number of iterations that a deterministic 
gradient descent method 
requires to attain $\|\nabla f(\theta^k)\| \leq \epsilon$ is on the order of $1/\epsilon^2$; this bound is known to be tight 
\cite{cartis2010complexity,cartis2012oracle}.\footnote{\change{In general, the worst-case iteration complexity for \emph{any} unconstrained deterministic first-order optimization method to attain $\|\nabla f(\theta^k)\|\leq\ \epsilon$ given a Lipschitz continuous gradient and Hessian is on the order of $\epsilon^{-1.75}$ \cite{carmonlower}. Such a theoretically best worst-case rate is indeed attained by accelerated gradient descent \cite{carmon2017convex,jin2018accelerated}.} }
In the same class of problems, the number of iterations for a first-order method based on the classical Robbins-Monro stochastic gradient iteration to achieve the 
same degree of stationarity is on the order of $1/\epsilon^4$ \cite{ghadimilan}, which is strictly worse.\footnote{\change{In general, the theoretical best worst-case iteration complexity for \emph{any} unconstrained stochastic first-order optimization method to 
attain $\|\nabla f(\theta^k)\|\leq\epsilon$ - the definition of which must be more carefully provided than in this footnote - is on the order of $1/\epsilon^3$ \cite{arjevani19}. Such a theoretical best rate is attained, for instance, by SPIDER \cite{fang2018spider}, at the cost of evaluating two stochastic gradients per iteration (as opposed to one stochastic gradient per iteration required by the classical Robbins-Monro iteration). }}
In the quantum computing setting, adopting an algorithm like SHOALS translates to significantly lower 
total time, since certain large parts of the latency costs are paid
per iteration, rather than per shot, at the expense of potentially higher (but dynamic) numbers of shots. 
We note that a recent paper on another classical optimizer for VQAs employs a Bayesian 
line search~\cite{tamiya2021stochastic}, which is fundamentally 
different from the line search considered in this paper. 

\change{
We stress that SHOALS is not a new algorithm, but is a practical extension of ALOE \cite{jinscheinbergxie}.
Specifically, SHOALS allows for per-parameter sampling of directional derivatives, enabled by parameter shift gradients. 
However, we do note that this is the first work to explicitly consider adaptive stochastic methods for use in VQAs, 
and we illustrate - in terms of latency cost and shot execution tradeoffs - why such algorithms are extremely promising 
in this setting, even though they were not explicitly designed for this setting.}

We provide numerical
experiments demonstrating the efficacy of SHOALS  on a variety of chemical molecules. 
Depending on the specifics of a simulated quantum computing environment, SHOALS can
reduce the time needed to reach chemical accuracy by several orders
of magnitude compared with other state-of-the-art optimization algorithms, 
such as iCANS~\cite{kubler2020adaptive} and Adam. 
In particular, our comparisons employ various measures of latency, 
showing how the effect of circuit switching and network latency costs
affect the comparisons. 
As suggested by the theory on which we elaborate in this paper, 
we find that SHOALS frequently outperforms popular methods based on 
the stochastic gradient iteration, especially in settings where latency cost is high compared with shot cost. 
Our algorithm and numerical experiments highlight the importance of making more efficient use of limited quantum
resources, allowing for further exploration of more interesting problems.

\subsection{Variational Quantum Eigensolver}

In the VQE, a standard 
example of a VQA~\cite{cerezo2021variational}, we seek to find an approximation
of the lowest eigenvalue of a Hermitian matrix, $H$, which is often known as the
Hamiltonian. VQE relies on what is known as the
``variational principle,'' which states
\begin{equation}\label{var_princ}
  \langle \psi(\mathbf{\theta}) | H | \psi(\mathbf{\theta}) \rangle \ge E_0,
\end{equation}
where $E_0$ is the true value of the lowest eigenvalue (often known as the
ground state energy) and $|\psi(\mathbf{\theta})\rangle$ is any parameterizable
function (typically known as an ansatz). The variational principle,
Eq.~\eqref{var_princ}, states that an upper bound for the true ground state
energy can be obtained by measuring the energy (that is, computing $\langle \psi(\mathbf{\theta}) | H |
\psi(\mathbf{\theta}) \rangle$) for any function, $|\psi(\mathbf{\theta})\rangle$.
A flexible, expressive wavefunction $|\psi(\mathbf{\theta})\rangle$ 
and efficient, effective optimization are key to realizing high-quality approximations
of the true ground state energy. Because of their natural ability to express
quantum mechanics, quantum computers are excellent candidates for preparing
various ansatzes.

For standard quantum chemistry problems, we can write the
Hamiltonian, $H$, after preprocessing through a fermion-to-spin 
transformation such as the Jordan--Wigner transformation~\cite{jordan1993paulische},
as a sum of terms,
\begin{equation} \label{eq:hamiltonian}
  H = \sum_{j=0}^N a_j P_j,
\end{equation}
where $a_j$ 
is a (typically real-valued) constant prefactor and $P_j$ is a Pauli string. 
Note that the $N$ here is often expressed as $\mathcal{O}(N_{o}^4),$ where $N_{o}$ is the 
number of orbitals used (and is often equal to the number of qubits).
Each Pauli string can be measured efficiently on a quantum
computer. The VQE problem can thus be
expressed as an optimization problem,
\begin{equation}
\label{eq:obj-braket-pre}
  \min_{\mathbf{\theta}} \sum_j^N a_j \langle \psi(\mathbf{\theta}) | P_j | \psi(\mathbf{\theta}) \rangle.
\end{equation}

Each of the terms in the sum in \eqref{eq:obj-braket-pre} is independently estimated  stochastically  by
repeatedly querying the quantum computer for some number of shots. 

For convenience of notation in discussing an optimization algorithm, we will, 
in the remainder of this paper, abuse notation and rewrite \eqref{eq:obj-braket-pre} as

\begin{equation}
\label{eq:obj-braket}
  \min_{\mathbf{\theta}} \sum_j^N a_j P_j(\theta) \equiv
  \min_{\mathbf{\theta}} f(\theta),
\end{equation}
where $P_j(\theta)$ denotes $\langle \psi(\mathbf{\theta}) | P_j | \psi(\mathbf{\theta}) \rangle$. 
By using the parameter shift rule~\cite{schuld2019evaluating},  partial derivatives 
for the types of ansatzes used in this work
can be calculated as
\begin{equation}
\label{eq:ps_grad}
    \partial_i f(\theta) = \frac{f(\theta+\frac{\pi}{2}e_i) - f(\theta-\frac{\pi}{2}e_i)}{2},
\end{equation}
where $e_i$ represents the unit vector in the direction of parameters $i$.
We utilize a Trotterized, unitary coupled cluster, singles doubles (UCCSD) ansatz~\cite{lee2018generalized,peruzzo2014variational} 
for this work, which is defined as
\begin{equation}
    |\psi(\theta)\rangle = \bigg(\prod_{i}^{d} \prod_{j} \exp{\theta_{i,j} \hat{t}_j} \bigg) |0\rangle,
\end{equation}
where $d$ represents what we term depth, which is the number of times to repeat the basic Trotterized 
UCCSD ansatz with unique parameters (similar to the $k$ parameter
in Ref.~\cite{lee2018generalized}), and $\hat{t}_j$ is either a singles 
($\hat{t}_j = \hat{a}^\dagger_k \hat{a}_l$)
or doubles 
($\hat{t}_j = \hat{a}^\dagger_k \hat{a}^\dagger_l \hat{a}_m \hat{a}_n$) cluster operator.
We note that our theory should extend to any ansatz, such as those generated in qubit coupled cluster~\cite{ryabinkin2018qubit}
or hardware-efficient ansatzes~\cite{kandala2017hardware},
and can be embedded within methods to generate more compact ansatzes, such as ADAPT-VQE~\cite{tang2021qubit} and 
unitary selective coupled cluster~\cite{fedorov2021unitary}.

\subsection{Single-Shot Estimators}
A gradient-based optimization algorithm for solving \eqref{eq:obj-braket} will require the ability to compute $f(\theta)$ and $\nabla f(\theta)$.
However, because each term $P_j(\theta)$ in \eqref{eq:obj-braket} is in fact a random variable for which a shot execution on a quantum computer gives an estimate valued in $\{-1,1\}$, we cannot directly compute $f(\theta)$. 
Instead, we obtain what we call a \emph{single-shot estimator} $f(\theta;\xi)$ of $f(\theta)$, which is computed by taking one observation of each Pauli string $P_j(\theta)$ and returning the linear combination specified by the prefactors $\{a_j\}$.  
We note that although we call $f(\theta;\xi)$ a \emph{single-shot} estimator, it does not cost one shot to obtain a single-shot estimator; rather, a single-shot estimator costs one shot per circuit necessary to evaluate every Pauli string, which is typically $O(N_q^4)$, where $N_q$ is the number of qubits used. Using operator grouping, one potentially can reduce this number to $O(N_q^3)$~\cite{gokhale2020n}. 
The random variable $\xi$ in $f(\theta;\xi)$ is intended to encapsulate all the randomness inherent in estimating $f(\theta)$ with a quantum computer. 

Analogously, using the parameter shift gradient, Eq.~\eqref{eq:ps_grad}, we can compute a single-shot estimator $g_i(\theta;\xi)$ of each partial derivative $\partial_i f(\theta)$. 
When we assemble the single-shot estimators of partial derivatives into a full gradient vector, we denote this latter quantity $g(\theta;\xi)$.

\section{Shot-Adaptive Line Search}
 
 Our method is fundamentally based on a stochastic line search method proposed over several papers \cite{paquette2020stochastic,berahas2021global} but is closest to one  referred to as Adaptive Linesearch with Oracle Estimations, or ALOE \cite{jinscheinbergxie}.
 To disambiguate from the use of the word ``oracle" in quantum computing, 
 we instead use the name SHOt-Adaptive Line Search, or SHOALS. 
We describe this method in Algorithm~\ref{alg:sls} and discuss key elements of the algorithm. 
 
 \begin{algorithm}
\caption{SHOt-Adaptive Line Search (SHOALS) \label{alg:sls}}
\textbf{Initialization:} Choose constants $\gamma>1,c\in(0,1),$ and $\alpha_{\max}>0$. \\
Choose accuracy level $\epsilon_f> 0$.\\
\textbf{Input:} initial point $\thetab^0\in\R^n$, initial step size $\alpha_0>0$. \\
\For{$k = 0,1,2,\ldots$}{
	Compute an estimate $g^k$ of $\nabla f(\thetab^k)$ (in practice; $g^k$ will be assembled coordinate-wise from averages of $N_{g_i,k}$ samples of $g_i(\thetab^k;\xi)$; see \eqref{eq:ngik}). \\
	Set a trial point $\sba^k\gets \thetab^k - \alpha_kg^k$. \\
	Compute function estimates $f^k_0$ and $f^k_\sba$ of $f(\thetab^k)$ and $f(\sba^k)$, respectively (in practice, $f^k_0$ and $f^k_s$ will be computed as averages of $N_{f,k}$ samples of $f(\thetab^k;\xi)$ and $f(\sba^k;\xi)$; see \eqref{eq:nfk}). \\ \label{line:epsf}
	\eIf{$f^k_\sba \leq f^k_0 - c\alpha_k\|g^k\|^2 + 2\epsilon_f$ \label{line:suffdec}}
	{
	$\thetab^{k+1}\gets \sba^k$\\
	$\alpha_{k+1}\gets \min\{\alpha_{\max},\gamma\alpha_k\}$\\ \label{line:stepsize_grow} 
	}
	{
	$\thetab^{k+1}\gets \thetab^k$\\
	$\alpha_{k+1}\gets\gamma^{-1}\alpha_k$\\
	}
	}
\end{algorithm}

We note that Algorithm~\ref{alg:sls} would be a classical line search method with an Armijo line search with the notable differences that 
\begin{itemize}
\item \lineref{stepsize_grow} would be replaced with $\alpha_k\gets 1$; 
\item the estimates $g^k, f_0^k,$ and $f_s^k$ would be replaced with the respective true values $\nabla f(\thetab^k), f(\thetab^k)$, and $f(\sba^k)$; and
\item $\epsilon_f$ would always be initialized to $0$. 
\end{itemize}

Under particular conditions on the estimates $g^k, f_0^k,$ and $f_s^k$,
as well as the determination of $\epsilon_f$ in \lineref{epsf}, Jin et al.~\cite{jinscheinbergxie} demonstrate a special kind of convergence result concerning the iterates of Algorithm~\ref{alg:sls}. 
We now elaborate on these conditions, which will be enforced in SHOALS. 

\subsection{Gradient Estimation}\label{sec:gradients} 
In our setting, a directional derivative estimate $g_i^k \approx \partial_i f(\theta^k)$ (see \eqref{eq:ps_grad}) is computed by averaging a number $N_{g_i,k}$ of single-shot estimates of the directional derivative, $g_i^k(\thetab^k;\xi)$. 
As a result, the gradient estimate $g^k$ can be written $g^k(\thetab^k; \xi^k)$, where $\xi^k$ denotes a realization of the random variable $\Xi^k$ representing all possible shot sequences that could be observed when obtaining $N_{g_i,k}$ many single-shot estimates of each $\partial_i f(\thetab^k)$, and concatenating the estimates into a $n$-dimensional gradient vector. 
We stress that in our notation, whenever $\xi$ has a superscript $(k)$, there is an implied (set of) sample size(s) $N_{g_i,k}$ to distinguish it from the atomic random variable $\xi$ involved in the single-shot estimator.

\change{We give reasoning in the appendix, but we ultimately require, for each $i=1,\dots,n$
\begin{equation}\label{eq:ngik}
N_{g_i,k} = 
 \left\lceil \displaystyle\frac{\mathbb{V}\left[ g_i(\thetab^k;\xi)\right]}{p\left(\max\{L_i\alpha_kg_i(\thetab^k;\xi^k),\epsilon_g\}\right)^2} \right\rceil,
 \end{equation}
for some $p\in(0,\frac{1}{2}^n)$ and some $\epsilon_g > 0$.
In \eqref{eq:ngik}, $L_i$ denotes a global Lipschitz constant of $\partial_i f(\theta)$
and $\mathbb{V}\left[g_i(\thetab^k;\xi)\right]$ denotes the variance of the single-shot estimator $g_i(\thetab^k;\xi)$.
Loosely speaking, enforcing \eqref{eq:ngik} will guarantee that, with high probability $(1-p)$, the error in the estimate $g_i(\theta^k;\xi^k)$ of $\partial_i f(\theta^k)$ will be bounded in such a way that that it does not interfere with the gradient descent dynamics.}

As in \cite{kubler2020adaptive}, we note that we can yield an upper bound on each $L_i$ 
due to the parameter shift rule and the expression of the cost function as a prefactor-weighted sum of Pauli matrices. 
In particular, 
given an expression for the non-mixed partial second derivative $\partial_i\partial_i f(\thetab^k)$, we can yield an upper bound on each $L_i$ by the sum of the absolute values of prefactors
obtained after applying a parameter shift \eqref{eq:ps_grad} twice.

To make practical use of \eqref{eq:ngik}, it remains to handle the unknown variance term in the numerator. 
The variance of the estimator $g_i(\thetab^k;\xi)$ is unknowable because computing a population variance requires knowing the expectation $\nabla f(\thetab^k)$. 
In our implementation we simply compute the sample variance of the sample of size $N_{g,k}$ before the end of the $k$th iteration to use as an estimate of the variance in the $(k+1)$st iteration. 
It is theoretically reasonable to believe that the variance $\mathbb{V}\left[ g(\thetab^k;\xi)\right]$ should not change too rapidly in $\thetab^k$, justifying this choice.
We have also verified this assumption empirically through  experiments.
To summarize, our practical sample size is 
\begin{equation}\label{eq:practical_ngik}
N_{g_i,k} = 
 \left\lceil \displaystyle\frac{s_{g_i,k}^2}{p\left(\max\{L_i\alpha_kg_i(\thetab^k;\xi^k),\epsilon_g\}\right)^2} \right\rceil,
 \end{equation}
 where $s_{g_i,k}^2$ is the sample variance estimate of $g_i(\theta^k;\xi)$ obtained in the previous iteration. 

\subsection{Function Value Estimation}\label{sec:fvals} 
Similarly to the gradient estimation, we average a number $N_{f,k}$ of single-shot estimates of $f(\thetab^k)$ and $f(\sba^k)$ to yield estimates $f(\thetab^k;\xi^k)$ and $f(\sba^k;\xi^k)$. 
\begin{equation}\label{eq:nfk}
N_{f,k} = \left\lceil\frac{\mathbb{V}\left[f(\theta^k;\xi)\right]}{p\epsilon_f^2}\right\rceil,
\end{equation}
where $\epsilon_f>0$ is a parameter, and $p\in(\frac{1}{2},1)$.
\change{Formal reasoning is given in the appendix, but intuitively, averaging $N_{f,k}$ many samples 
guarantees with high probability that the error incurred by $f(\theta^k;\xi^k)$ is bounded by a 
constant $\epsilon_f$ that dictates the accuracy to which we would intend to reduce the optimality gap of our final solution.}

In our practical implementation, to avoid oversampling, we additionally employ a criterion resembling past work and the criterion in \eqref{eq:practical_ngik} (see \cite{storm, paquette2020stochastic}) and compute
\begin{equation}\label{eq:practical_nfk}
N_{f,k} = \min\left\{ \left\lceil\displaystyle\frac{s^2_{f,k}}{p(\alpha_k^2\|g(\thetab^{k};\xi^{k})\|^2)^2}\right\rceil, \left\lceil\frac{s^2_{f,k}}{\epsilon_f^2}\right\rceil\right\},
\end{equation}
where $s^2_{f,k}$ is the sample variance estimate of $\mathbb{V}[f(\thetab^k;\xi)]$, which would have been computed in the previous iteration. 

\subsection{Theoretical Results concerning SHOALS}\label{sec:theory}
In \cite{jinscheinbergxie}, a more formal statement of the following result is proven. 
\begin{theorem}\label{bigtheorem}
Suppose $f$ is bounded below. that is, there exists $f_*=f(\thetab^*)$ so that $f_*\leq f(\thetab)$ for all $\thetab$, and suppose $\nabla f(\theta)$ is Lipschitz continuous with constant $L_g$. 
Fix $\epsilon_g>0$ and $\epsilon_f>0$. 
Suppose the estimators $g(\thetab^k;\xi^k)$ and $\{f(\thetab^k;\xi^k),f(\sba^k;\xi^k)\}$ employed in each iteration of Algorithm~\ref{alg:sls} satisfy Condition~\ref{cond:g} and Condition~\ref{cond:f} with $\epsilon_g$ and $\epsilon_f$, respectively.
Let
$$\epsilon \geq 4\max\{\epsilon_g, (1+L_g\alpha_{\max})\sqrt{3L_g\epsilon_f}\}.$$
Then, with high probability, 
the number of iterations of Algorithm~\ref{alg:sls}, $T_\epsilon$, required to attain
\begin{equation}
\label{eq:good_bounds}
    \|\nabla f(\thetab^{T_\epsilon})\| < \epsilon
\end{equation}
satisfies 
\begin{equation}
\label{eq:good_iter_complexity}
T_{\epsilon}\in\mathcal{O}\left(\displaystyle\frac{L_g^2\big(f(\thetab^0)-f(\thetab^*)\big)}{\epsilon^2}\right).
\end{equation}
The probability in this result is with respect to the $\sigma$-algebra of the realizations of $\xi^k$ observed over the run of the algorithm. 
\end{theorem}
We remark on some important aspects of Theorem~\ref{bigtheorem}. 
First, as indicated by the big-$\mathcal{O}$ notation, this result provides a \emph{worst-case} guarantee. 
Virtually never in practice is a stopping time on the order of $1/\epsilon^2$ many iterations realized.
However, the result of Theorem~\ref{bigtheorem} recovers 
up to small constants the worst-case iteration complexity result for \emph{deterministic} gradient descent algorithms, which is known to be tight \cite{cartis2010complexity}. 

\subsection{Comparison with Stochastic Gradient Methods}\label{sec:comparisontheory} 
The result in Theorem~\ref{bigtheorem} should be contrasted with known results for the worst-case performance of stochastic gradient methods.
We provide a prototypical stochastic gradient method in \algref{sgd}.
\begin{algorithm}
\caption{Generic Stochastic Gradient Method \label{alg:sgd}}
\textbf{Initialization:} Choose step size $\alpha$, initial point $\thetab^0\in\R^n$ \\
\For{$k = 0,1,2,\ldots$}{
	Compute an estimate $g^k$ of $\nabla f(\thetab^k)$. \\
	$\thetab^{k+1}\gets \thetab^k-\alpha g^k$. \\
	}
\end{algorithm}
Results in Corollary 2.2 of Ref.~\cite{ghadimilan} and Theorem 4.8 of Ref.~\cite{bottoucurtisnocedal} show essentially that in order to guarantee 
\begin{equation}\label{eq:bad_bounds} 
\mathbb{E}\left[\displaystyle\frac{1}{T_{\epsilon}}\displaystyle\sum_{k=0}^{T_{\epsilon}} \|\nabla f(\thetab^k)\| \right] \leq\epsilon,
\end{equation}
one must run a stochastic gradient method for 
\begin{equation}
\label{eq:bad_iter_complexity}
T_\epsilon\in\mathcal{O}\left(\displaystyle\frac{L_g(f(\thetab^0)-f(\thetab^*))}{\epsilon^4}\right)
\end{equation}
many iterations, provided the step size $\alpha$ is chosen sufficiently small, typically satisfying a bound like  $\alpha\leq \frac{1}{L_g}$.
For ease of presentation, we deliberately ignored in \eqref{eq:bad_bounds} and \eqref{eq:bad_iter_complexity} an additional dependence on the variance of $g^k$, which decreases linearly with the sample size or ``batch size" $b$.
This is a significant oversight because in the worst case one requires $b\in\Theta(T_{\epsilon})$, where $T_{\epsilon}$ is as in \eqref{eq:bad_iter_complexity} in order for theoretical results such as \eqref{eq:bad_bounds} to hold. 
This is certainly  concerning in settings where accuracy matters (e.g., chemical accuracy in the problems we are considering), but in many machine learning applications it is considerably less concerning. 
In light of this intentional oversight, we remark that the worst-case results we are presenting for SG methods are, in fact, erring on the side of optimism. 

The bound in \eqref{eq:bad_bounds} can be interpreted as saying that, given a fixed number of iterations of $T$, the \emph{expected value} (over the $\sigma$-algebra of all sources of randomness) of the average value (averaged over $k=1,\dots,T$) of $\|\nabla f(\thetab^k)\|$ is bounded by $\epsilon$.
As a corollary, one can show that if one selects uniformly at random $R\in\{1,\dots,T_{\epsilon}\}$, then
\begin{equation}\label{eq:bad_bounds_random} 
\mathbb{E}\left[ \|\nabla f(\thetab^R)\| \right] \leq\epsilon.
\end{equation}
These results have been shown to be tight in the general case~\cite{drori2020complexity} and improvable to $\mathcal{O}\left(\frac{1}{\epsilon^{3.5}}\right)$
in special cases \cite{fang2019sharp}. 
Popular variants of SG algorithms like \algref{sgd}, such as Adam, have been shown to have virtually the same type and quality of error bounds as that in \eqref{eq:bad_bounds}; see Theorem 1 of Ref.~\cite{reddi2018adaptive}. 
Methods such as iCANS~\cite{kubler2020adaptive} are also based on the stochastic gradient iteration and thus are susceptible to the same worst-case performance. 

The trade-off between a $\mathcal{O}\left(\frac{1}{\epsilon^2}\right)$ worst-case complexity for algorithms of the class in Algorithm~\ref{alg:sls} and the $\mathcal{O}\left(\frac{1}{\epsilon^4}\right)$ worst-case complexity for algorithms of the class in \algref{sgd} is that the estimates in Algorithm~\ref{alg:sls} require more samples (or, in the context of our parameter shift problems, shots) per iteration to satisfy Condition~\ref{cond:g} and Condition~\ref{cond:f}. 
The appropriate metric for comparing these algorithms is measured by the \emph{total work}, namely, a function of latency and shot acquisition times. 

We consider the following simplified setting.
We note that similar analyses (see \cite{bottou2007tradeoffs} and \cite[Section 4.4]{bottoucurtisnocedal})   have been done  to explain the apparent empirical dominance of stochastic gradient methods in empirical risk minimization settings, but the conclusions we will draw are very different because the quantum computing setting is very different from general machine learning settings in terms of overhead costs. 
We denote the shot acquisition time by $c_1$, the latency time of circuit switching by $c_2$, and the communication cost time by $c_3$. 
\change{All times vary from architecture to architecture. 
Shot acquisition time includes the time to reset the circuit to the initial state, 
the time to run the circuit, and  the time to record a single measurement, which can 
from a around 100 $\mu s$ in superconducting qubits~\cite{karalekas2020quantum}
and around 200 ms in neutral atom systems~\cite{briegel2000quantum}. 
Circuit switching time includes the time needed to compile a circuit and then to load it
on to the quantum computer's control system.
For superconducting qubits, compilation can take 200 ms, while interacting with the 
arbitrary waveform generator can take around 25 ms~\cite{karalekas2020quantum}.
Communication cost includes the time needed to communicate between the computer running the 
classical optimization algorithm and the actual quantum computer. This can range from nearly 
zero when the two devices are in the same room to many seconds when instructions need to be
sent over the Internet~\cite{sung2020using}.}
We note that every iteration of \algref{sls} requires two communications between the quantum and classical devices: one communication to 
obtain $g(\thetab^k;\xi^k)$ and a second communication to obtain estimates $f(\thetab^k;\xi^k)$ and $f(\sba^k,\xi^k)$.
Every iteration of \algref{sgd} requires only one communication to obtain $g(\thetab^k;\xi^k)$. 
Thus, denoting by $t^f_{\circ}$ and $t^g_{\circ}$ the number of circuits required to evaluate the one-shot estimators $f(\thetab^k;\xi)$ and $g(\thetab^k;\xi)$, respectively, we can summarize the overhead \emph{per-iteration latency cost} of each iteration as in Table~\ref{table:costs}. 

\begin{table*}[t]

\begin{tabular}{lcrccc}
     & Iterations   & $\times$ & $($Per-Iteration Latency Cost    & $+$ & Per-Iteration Shots Cost$)$\\\hline
SHOALS & $\mathcal{O}\left(\displaystyle\frac{L_g^2(f(\thetab^0)-f(\thetab^*))}{\epsilon^2}\right)$ & & $c_2 (2t^f_{\circ} +  t^g_{\circ}) + 2c_3$         & & $c_1\left(\displaystyle\frac{t^g_{\circ}}{\epsilon^2} + \frac{2t^f_{\circ}}{\epsilon^4}\right)$ \\
$=$ & \multicolumn{5}{l}{$\mathcal{O}\left(\displaystyle\frac{L_g^2(f(\thetab^0)-f(\thetab^*))\left[c_2(2t_{\circ}^f + t_{\circ}^g) + 2c_3\right]}{\epsilon^2} + 
\displaystyle\frac{L_g^2(f(\thetab^0)-f(\thetab^*))c_1t_{\circ}^g}{\epsilon^4} + 
\displaystyle\frac{2c_1t_{\circ}^fL_g^2(f(\thetab^0)-f(\thetab^*))}{\epsilon^6}\right)$}\\\hline 
SG & $\mathcal{O}\left(\displaystyle\frac{L_g(f(\thetab^0)-f(\thetab^*))}{\epsilon^4}\right)$ & & $c_2 t^g_{\circ} + c_3$ & & $bc_1t^g_{\circ}$\\
$=$ & \multicolumn{5}{l}{$\mathcal{O}\left( \displaystyle\frac{\left[bc_1t_{\circ}^g + c_2t_{\circ}^g + c_3\right]L_g(f(\thetab^0)-f(\thetab^*))}{\epsilon^4}\right)$} \\ \hline
\end{tabular}
\caption{\label{table:costs} Summary of total worst-case costs of SHOALS and a stochastic gradient (SG) method to achieve $\epsilon$-stationarity. The total worst-case cost is computed as the number of iterations to achieve $\epsilon$-stationarity multiplied by the sum of the per iteration costs. We recall that $t^f_{\circ}$ and $t^g_{\circ}$ denote the number of circuits needed to compute $f(\thetab^k;\xi)$ and $g(\thetab^k;\xi)$, respectively. For an SG method, we denote a constant batch size parameter by $b$.}
\end{table*}

Now, denote by $t^f_s$ and $t^g_s$ the number of shots requested per iteration. 
By Theorem~\ref{bigtheorem}, if we choose $\epsilon = \epsilon_g < 1$
and $\epsilon_f = \epsilon^2$, then in the worst case, and up to constants,
\algref{sls} will require on each iteration $\frac{1}{\epsilon^2}$ samples of $g(\thetab^k;\xi)$
and $\frac{1}{\epsilon^4}$ samples of $f(\thetab^k;\xi)$, yielding
$t^g_s \leq \frac{t^g_{\circ}}{\epsilon^2}$ and $t^f_s \leq \frac{2t^f_{\circ}}{\epsilon^4}.$
On the other hand, each iteration of \algref{sgd} will  require only a constant number of samples, $b$, of $g(\thetab^k;\xi)$.
That is, without loss of generality, and up to constants, $t^g_s = bt^g_{\circ}$.
All of these costs are summarized in Table~\ref{table:costs}. 

We note that SHOALS  incurs only the communication cost, $c_3$, on the order of $1/\epsilon^2$ many times, whereas the cost is incurred on the order of $1/\epsilon^4$ many times for SG methods. 
Additionally, the latency switching cost $c_2$ is  present only on the $1/\epsilon^2$ term for SHOALS, whereas it appears on the order of $1/\epsilon^4$ many times in SG methods. 
On the other hand, we note that SHOALS has an unfortunate $1/\epsilon^6$ dependence; however, if $c_1$ is small---that is, if the cost of shot acquisition is small relative to the total latency and communication costs $c_3 + c_2(t^f_{\circ} + t^g_{\circ})$)---then the $1/\epsilon^6$ term may not dominate. 

To quantify this trade-off a bit more, 
we make simplifying assumptions that $t_{\circ}^g = 2n t_{\circ}^f$, 
where $n$ is the dimension of the parameter vector $\thetab$, which is likely correct up to constants in the parameter shift setting~\cite{schuld2019evaluating}, 
and that the batchsize employed by the SG method is $b=1$, which is correct up to constants in practice.
Then, in this setting and per Table~\ref{table:costs},
one can derive that
a sufficient condition for the 
total cost of SHOALS being less than the total cost of an SG method is the satisfaction of 
\begin{equation}\label{eq:simple_rule}
c_1 
< \displaystyle\frac{\epsilon^2(2nt_{\circ}^fc_2 + c_3)}{2t_{\circ}^fL_g},
\end{equation}
still assuming without loss of generality that $L_g > 1$.
That is, roughly speaking, if the cost of obtaining a single shot is smaller by a factor of $\epsilon^2$ than the overhead latency cost of one iteration of an SG method, normalized by both the number of circuits needed to evaluate $f(\thetab^k;\xi)$ and the gradient Lipschitz constant, then SHOALS is expected to yield better worst-case guarantees in terms of total costs than an SG method does. 
Thinking to the future of the deployment of VQAs, as $n$ becomes large,
$L_g$ ought to increase linearly in $N$ in \eqref{eq:obj-braket} and hence linearly in $n$; therefore,
\eqref{eq:simple_rule} asympotically amounts to 
$c_1 < \epsilon^2c_2$. 

In our experiments we will examine this trade-off empirically; but to summarize, these theoretical insights suggest that, asymptotically, when shot acquisition time is less than a factor of $\epsilon^2$ times the latency switching time, we expect SHOALS to outperform methods based on the iteration in Algorithm~\ref{alg:sgd}. 

\section{Numerical Results}
We note that, similar to stochastic gradient methods, SHOALS is remarkably simple to implement. 
It has also been documented in past work that adaptive methods like SHOALS are far less sensitive to changes in hyperparameters (such as step sizes) than stochastic gradient methods are, virtually eliminating the need for tuning hyperparameters. 
In our implementation of SHOALS, we chose $\gamma = 2, c = 0.2, \alpha_{\max} = 1$, and $\alpha_0 = 1$.

 In our implementation, when computing sample sizes via \eqref{eq:practical_ngik} and \eqref{eq:practical_nfk} we set $p=0.1$. 
 Motivated by Theorem~\ref{bigtheorem}, we ignore constant terms and set $\epsilon_g = \sqrt{\epsilon_f}$, where in turn
 motivated by chemical accuracy \cite{helgaker2014molecular}, we set $\epsilon_f = 0.0016$. 
 We note that because $\epsilon\in\Theta(\epsilon_g)$ per Theorem~\ref{bigtheorem}, this essentially means that the factor of $\epsilon^2$ that played a critical role in the end of our discussion in Section~\ref{sec:comparisontheory} is effectively $\epsilon_f$.  
 \begin{figure*}[t]
    \centering
    \includegraphics[width=.999\textwidth]{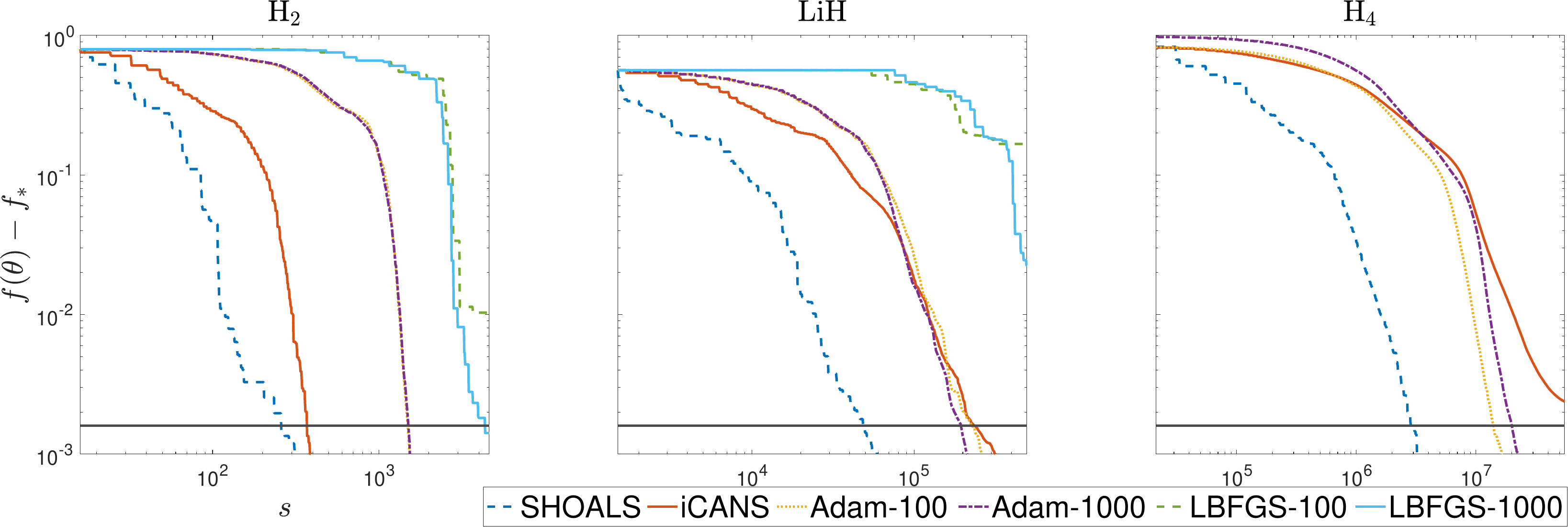}
    \caption{Comparison of SHOALS, iCANS, Adam, and LBFGS. Numbers next to Adam and LBFGS denote the shot size employed for each circuit. The $x$-axis shows the cost of the optimization in time in the simulated superconducting environment, while the $y$-axis shows the distance of the value of $f(\thetab)$ to the known global minimum $f_*$ after the corresponding number of budget units is spent on an optimizer. The solid horizontal line denotes chemical accuracy. For each solver, we show median performance over 30 runs with randomly generated initial points $\thetab^0$ (common to all solvers) with the line indicated in the legend. We remark again that LBFGS made negligible progress on H$_4$, and hence it is omitted  from the rightmost figure. \label{fig:median_trajectories}}
\end{figure*}

We first consider a hypothetical, forward-looking, superconducting hardware environment outlined in Ref.~\cite{sung2020using}; in particular, we set $c_1 = 1.0\times 10^{-5}$ s, $c_2 = 0.1$ s, and $c_3 = 4.0$ s. 
 Before showing any results, we remark that in this setting, $c_1 < \epsilon_f c_2$; and so, from our discussion in Section~\ref{sec:comparisontheory}, the theory suggests that in worst-case settings and for larger $n$, we expect SHOALS to outperform all of the stochastic gradient methods. 
Using these hardware times, we will display in our figures the total simulated time on the $x$-axis, where $t^f_s$ and $t^g_s$ (the numbers of shots requested by a method per iteration) are the actual values requested in the simulations. 

We consider three quantum chemistry problems, for which we give relevant statistics in Table~\ref{table:problems}.
For all molecules we use parity fermion-to-spin mapping with two-qubit reduction~\cite{bravyi2017tapering} to reduce the number of qubits,
as implemented in Qiskit~\cite{Qiskit}.
We simulate H$_2$ with two qubits at an equilibrium distance of 0.74~$\AA$. For LiH, we use an
equilibrium distance of 1.595~$\AA$, and we freeze core orbitals (as in~\cite{kandala2017hardware}),
leading to only four qubits. We use a square geometry for H$_4$, with each side having equilibrium length 1.23~$\AA$~\cite{lee2018generalized},
and comprising six qubits. 

\begin{figure*}[t]
    \centering
    \includegraphics[width=.999\textwidth]{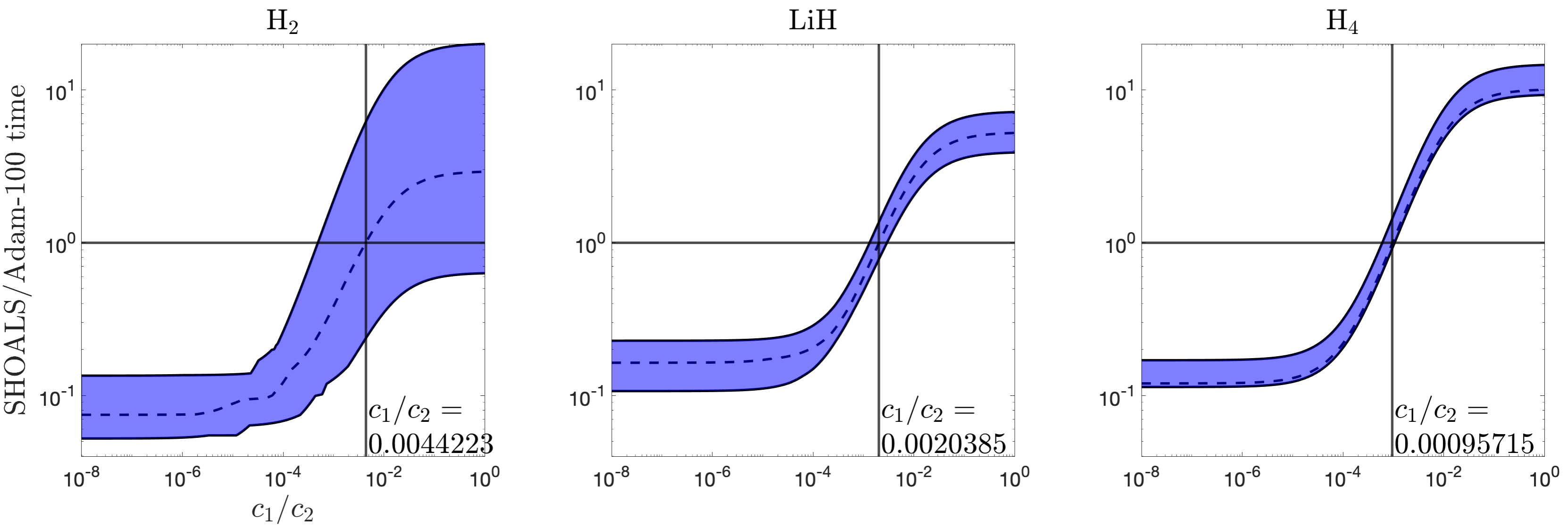}
    \caption{For each of the tested molecules, we display the ratio of the wall-clock time needed for SHOALS to achieve chemical accuracy over the wall-clock time needed for Adam-100 to achieve chemical accuracy. 
    The dashed line denotes the median value of the ratio, while the transparent band shows the $25$th through $75$th percentiles of the ratio. 
    We also display in each plot the value of $c_1/c_2$ where the median ratio equals 1, that is, where the wall-clock time of the two solvers is equal. 
    \label{fig:varyc1c2}}
\end{figure*}
\begin{table}[]
\begin{tabular}{llll}
Ansatz       & n & $t_{\circ}^f$ & $t_{\circ}^g$ \\ \hline
H$_2$, UCCSD, Depth 1  & 3 & 5       & 40      \\
LiH, UCCSD, Depth 2 & 16 & 104      & 8320  \\
H$_4$, UCCSD, Depth 2 & 70 & 85 & 106080 \\ 
\end{tabular}
\caption{\label{table:problems} Relevant statistics for tested problems. As in preceding sections, $n$ denotes the number of parameters in the ansatz, $t_{\circ}^f$ denotes the number of circuits that must be evaluated to obtain an evalution of $f(\thetab^k;\xi)$, and $t_{\circ}^g$ denotes the number of circuits that must be evaluated to obtain an evaluation of $g(\thetab^k; \xi)$.}
\end{table}

\subsection{Comparison of SHOALS with Other Solvers} 
We first compare the performance of SHOALS with two stochastic gradient methods: an implementation of iCANS and an implementation of Adam. For iCANS, we used the same hyperparameter settings described in~\cite{kubler2020adaptive}, with a minimum sample size set equal to 30, and  implemented what the authors referred to as ``iCANS1." 
For Adam, we employed a constant step size of $1.0/L_g$, momentum hyperparameters of  $0.9$ and $0.999$, and a denominator offset hyperparameter of $10^{-8}$. 
These settings are standard in practice.
We experimented with a batch size (in the quantum setting, average of a number of single-shot estimators $g(\thetab^k;\xi)$) of both 100 and 1000.
We remark that one could tune all of the hyperparameters of both iCANS and Adam to yield variations in practical performance, but one intention of these experiments is to demonstrate that \emph{standard} settings of SHOALS' hyperparameters yield relatively good performance compared with \emph{standard} settings of other solvers. 
Furthermore, we remark again that in actual NISQ settings, one would not want to expend budgets performing hyperparameter tuning. 

Additionally, we compare the performance of SHOALS with a deterministic first-order method, an implementation of L-BFGS \cite{zhu1997algorithm}.
This implementation is the standard one provided in Scipy and Qiskit and was \emph{not} designed to handle stochasticity in the function evaluation. Therefore, such a comparison is not particularly fair, since an L-BFGS optimizer has no guarantee of first-order convergence in such a setting. 
We experiment with L-BFGS employing both a batch size of 100 and 1000. 
We remark that there is ongoing work on the development of L-BFGS optimizers capable of handling stochastic estimates of function and gradient values \cite{bollapragada2018adaptive,bollapragada2018progressive}. 
In fact, Pacquette and Scheinberg~\cite{paquette2020stochastic} discuss how to incorporate approximate second-order information, such as that provided by L-BFGS techniques, into the adaptive line search framework
We leave the incorporation of such second-order information into SHOALS as future work, since it is not the focus of this current paper.

Results are illustrated in Figure~\ref{fig:median_trajectories}, and additional statistics are given in Table~\ref{table:summary}. 
We see that, in this superconducting setting, the median wall-clock time performance of SHOALS is consistently better than any of the compared stochastic gradient and L-BFGS methods. 
As expected, on all of the tested problems, the median performance of L-BFGS with a fixed batchsize fails to attain chemical accuracy. For this reason, we chose not to present that solver for the largest tested molecule, H$_4$.

In light of our  discussion in Section~\ref{sec:comparisontheory} and the inherent uncertainty in the specific timings of the superconducting setting, we additionally present in Figure~\ref{fig:varyc1c2}, for the same molecules, the median (and first and third quantiles) of the ratio of the wall-clock time needed for SHOALS to attain chemical accuracy to the wall-clock time needed for Adam-100 to attain chemical accuracy when we fix $c_3 = 0$ (i.e., we completely disregard cloud communication time) and vary the ratio of $c_1/c_2$.

As expected by our discussion in Section~\ref{sec:comparisontheory}, for each tested molecules, as $c_1/c_2$ tends to 0, the advantage of using SHOALS over a stochastic gradient method is clear. 
What merits further investigation, but requires testing significantly larger molecules, is whether  the trend that the value of $c_1/c_2$ at the breakeven point is decreasing with $n$.
One explanation would be that our predictions are based on \emph{worst-case} behavior and the optimization problems we are testing are certainly not exhibiting worst-case behavior. 
In any case, it would be foolhardy to extrapolate any trend from three datapoints; testing molecules that require orders of magnitudes more circuits to compute a single gradient will involve an extensive computational campaign, which is a subject for future work.

\section{Conclusions and Future Work}
In this paper we demonstrate a new optimization algorithm, SHOt-Adaptive Line Search
(SHOALS). We provide theoretical arguments and numerical evidence that SHOALS should 
outperform other state-of-the-art optimization algorithms when the shot acquisition
time is much less than the circuit switching time.
This is because the dynamics of the deterministic line-search method that SHOALS stochastically approximates requires fewer iterations but more shots per iteration,
to reach a given level of stationarity.
Standard stochastic gradient methods,
on the other hand, can utilize smaller numbers of shots per iteration to make expected
improvement in the cost function. 
The latency regime in which SHOALS is expected to perform better than stochastic gradient methods is realized by today's
superconducting qubit quantum computers~\cite{sung2020using}. 
However, when
the circuit switching time and shot acquisition time are closer in magnitude, 
such as in trapped ion systems~\cite{clark2021engineering},
we show, theoretically and numerically, that stochastic gradient 
methods are expected to outperform SHOALS. 
Because quantum devices have a non-negligible per-iteration
overhead (from the need to switch between many circuits), these latency considerations 
are important in determining which algorithm to use in variational quantum algorithms,
whereas in optimization of classical functions, such as those encountered in machine learning, these considerations are considerably less important. 

In future work we intend to investigate the possibility of tighter worst-case complexities for SHOALS, in terms of the wall-clock time metric, than those presented in Table~\ref{table:costs}. 
Indeed, as a simplifying assumption, we assumed that SHOALS required $\mathcal{O}(1/\epsilon^2)$ many samples of $g(\thetab^k;\xi)$ and $\mathcal{O}(1/\epsilon^4)$ many samples of $f(\thetab^k;\xi)$, $f(\sba^k;\xi)$ per iteration, but this is in fact  necessary only when $k$ is near the stopping time $T_\epsilon$. 
A more rigorous analysis would consider the \emph{total work} of SHOALS, such as in the analysis performed in Ref.~\cite{pasupathy2018sampling}.

We also intend to pursue multiple directions in the practical implementation of SHOALS.
In one direction, we are interested in the
incorporation of (approximate) second-order information into SHOALS in order to further decrease the iteration complexity and hence relax further the total accumulated per-iteration latency switching costs.
As suggested previously, we are inspired by work in 
Refs.~\cite{bollapragada2018adaptive,bollapragada2018progressive,paquette2020stochastic}.
We also intend to investigate an operator-sampling variant of SHOALS, in the style of Ref.~\cite{arrasmith2020operator}, in order to further reduce both latency and shot acquisition costs. 
In particular, we may choose on certain iterations, based on sample variance estimates, to evaluate only a subset of Pauli strings in \eqref{eq:obj-braket}, or perhaps to  evaluate only a subset of directional derivative estimates $g_i(\thetab^k;\xi^k)$ when forming a gradient estimate $g^k$. 
Such approaches would resemble doubly stochastic coordinate descent methods, which have been  studied in the optimization literature \cite{xu2015block}. 

Moreover, and as suggested in our numerical experiments, we intend to conduct a larger computational campaign to look at larger molecules in order to ascertain our belief in the asymptotic performance of SHOALS. 

\begin{section}{Acknowledgements}
This work was supported in part by the U.S.\ Department of Energy (DOE), Office of Science, Office of Advanced Scientific Computing Research Applied Mathematics activity and AIDE-QC project under contract number DE-AC02-06CH11357.
\end{section}

\bibliography{papers} 

\appendix
\begin{section}{Appendix}
\change{
\subsection{Derivation of \eqref{eq:ngik}}
We require $g^k(\thetab^k;\xi^k)$ to satisfy the following property.
\begin{condition}\label{cond:g}
For some $p'\in(0,1/2)$ and for some accuracy level $\epsilon_g>0$,
\begin{align}
    \mathbb{P}_{\xi^k}\left[e_g(\thetab^k;\xi^k)\leq \max\{\epsilon_g,L_g\alpha_k\|g(\thetab^k;\xi^k)\|\right\}] \\ \nonumber
     \geq 1-p',
\end{align} 
where $e_g(\thetab^k;\xi^k)$ denotes $\|g(\thetab^k; \xi^k)-\nabla f(\thetab^k)\|$ and 
$L_g$ denotes a local Lipschitz constant of the gradient $\nabla f(\thetab^k)$. 
\end{condition} 
In other words, Condition~\ref{cond:g} requires that, with sufficiently high probability, the error in the estimator be bounded by the Lipschitz constant times $\|\sba^k-\thetab^k\|$, an upper bound on the change in the gradient size. 

Applying Chebyshev's inequality, we can guarantee a gradient estimate $g^k$ satisfying Condition~\ref{cond:g} provided we choose each $N_{g_i,k}$ according to \eqref{eq:ngik}.  
Indeed, \eqref{eq:ngik} is the number of independent observations needed to attain
\begin{align*}
\mathbb{P}_{\xi^k}\left[e_{g_i}(\theta^k;\xi^k) \leq \max\{\epsilon_g, L_i\alpha_k|g_i(\theta^k;\xi^k)|\}\right]& \\
\geq & 1-p,\\
\end{align*}
where $e_{g_i}(\theta^k;\xi^k) = |g_i(\theta^k;\xi^k) - \partial_i f(\theta^k)|$.
Then, by a union bound, with probability $1-p^n \geq 1-p'$, 
$$e_g(\theta^k;\xi^k) \leq \max\{\epsilon_g, L_g\alpha_k\|g(\theta^k;\xi^k)\|\}. $$

\subsection{Deriving $N_{f,k}$}  
We now state our condition on $f(\thetab^k;\xi^k)$.
The condition for $f(\sba^k;\xi^k)$ is analogous. 
\begin{condition}\label{cond:f}
Define the \emph{estimation error} as
$e(\thetab^k;\xi^k):=|f(\thetab^k)-f(\thetab^k;\xi^k)|$.

The estimation error is a one-sided subexponential-like random variable with mean bounded by a constant $\epsilon_f>0$. 
That is, there exist parameters $(\nu,\beta)$ such that 
\begin{align}
\mathbb{E}_{\xi^k}\left[e(\thetab^k;\xi^k)\right] \leq \epsilon_f
\label{eq:fcond1}
\end{align}
and
\begin{align}
\mathbb{E}_{\xi^k}\left[\exp\left\{\lambda\left( e(\thetab^k;\xi^k) - \mathbb{E}_{\xi^k}\left[e(\thetab^k;\xi^k)\right] \right)\right\} \right]\nonumber \\ 
\leq \exp\left( \displaystyle\frac{\lambda^2\nu^2}{2}\right),
\quad \forall \lambda\in \left[0,\frac{1}{\beta}\right]. 
\label{eq:fcond2}
\end{align}
When \eqref{eq:fcond1} and \eqref{eq:fcond2} are satisfied, we say that \emph{$e(\thetab^k;\xi^k)$ is $(\nu,\beta)$-subexponential}. 
\end{condition}

The following corollary demonstrates that, for any $\epsilon_f>0$,  a sufficiently large sample size exists such that Condition~\ref{cond:f} holds. 

\begin{corollary}\label{cor:samplesize}
For any given 
$\epsilon_f>0$
we may choose
a sample size $N_{f,k}$ such that
$$\epsilon_f = \sqrt{\displaystyle\frac{V}{N_{f,k}}}$$ 
(where $V$ is as in Proposition~\ref{prop:one}) in order to guarantee
$e(\thetab^k;\xi^k)$ is $(\nu,\beta)$-subexponential with
\begin{align}
\nu=\beta=8\exp(2)\max\left\{\displaystyle\sqrt{\frac{2\mathbb{V}_{\xi}[e(\thetab^k;\xi)]}{N_{f,k}}},2(C+\epsilon_f)\right\},
\end{align}
and
$$\mathbb{E}_{\xi^k}[e(\thetab^k;\xi^k)]\leq \epsilon_f,$$
where $C$ is as in Proposition~\ref{prop:two}. 
\end{corollary}

We now prove Corollary~\ref{cor:samplesize}. }
We first record a result that was proven in \cite[Proposition 1]{jinscheinbergxie}.

\begin{proposition}\label{prop:one}
Suppose $e(\thetab^k;\xi)$ is a $(\hat\nu,\hat\beta)$-subexponential random variable
and $\mathbb{V}_{\xi}[f(\thetab^k;\xi)]\leq V$. 
Then
\begin{align*}
    \mathbb{E}_{\xi^k}[e(\thetab^k;\xi^k)] 
\leq 
\sqrt{\displaystyle\frac{V}{N_{f,k}}}
\end{align*}
and
$e(\thetab^k;\xi^k)$ is $(\nu,\beta)$-subexponential,
with $\nu=\beta=8\exp(2)\max\left\{\displaystyle\frac{\hat\nu}{\sqrt{N_{f,k}}},\hat\beta \right\}.$
\end{proposition}

Because of the form of the single-shot estimators $f(\thetab^k;\xi)$ in the parameter shift setting (a finite weighted sum of Bernoulli random variables), we have that $|e(\thetab^k;\xi)|\leq C$ deterministically,
where $C = 2\displaystyle\sum_i |a_i|.$
Thus, we can show the following result.

\begin{proposition}\label{prop:two}
$e(\thetab^k;\xi)$ is a $(\sqrt{2\mathbb{V}\left[e(\thetab^k;\xi)\right]},2(C+\epsilon_f))$-subexponential random variable. 
\end{proposition} 
\begin{proof}
Let $Y = \displaystyle\frac{e(\thetab^k;\xi) - \mathbb{E}\left[e(\thetab^k;\xi)\right]}{C+\epsilon_f}$, and suppose (as in Condition~\ref{cond:f}) that $\mathbb{E}\left[e(\thetab^k;\xi)\right] \leq \epsilon_f$. 
By definition of variance, $\mathbb{E}\left[Y^2\right] = \frac{1}{(C+\epsilon_f)^2}\mathbb{V}\left[e(\thetab^k;\xi)\right] := V. $
Moreover, 
$$|Y| \leq \displaystyle\frac{|e(\thetab^k;\xi)| + \mathbb{E}\left[e(\thetab^k;\xi)\right]}{C+\epsilon_f} \leq 1.$$ 
Thus, bounding the moment generating function, 
\begin{align*}
\mathbb{E}\left[\exp(\lambda Y)\right] = \mathbb{E}\left[\displaystyle\sum_{k=0}^\infty \frac{\lambda^kY^k}{k!} \right] 
\leq 1 + \displaystyle\sum_{k=2}^\infty \frac{\lambda^k}{k!} V \\
\leq 1 + \left(\displaystyle\sum_{k=0}^\infty \lambda^k \right)\frac{\lambda^2V}{2}.
\end{align*}
For all $\lambda < 1$, the geometric series on the right will converge; moreover, for $\lambda \leq \frac{1}{2}$, 
$$\mathbb{E}\left[\exp(\lambda Y)\right] \leq \exp\left(\displaystyle\frac{\lambda^2(\sqrt{2V})^2}{2} \right).$$
That is, $\frac{e(\thetab^k;\xi)}{C+\epsilon_f}$ is a $(\sqrt{2V},2)$-subexponential random variable. 
The result follows.
\end{proof}

Because the variance $\mathbb{V}_{\xi}\left[f(\thetab^k;\xi)\right]$ is bounded in the parameter-shift setting, ($f(\thetab^k;\xi)$ can be expressed as a finite weighted sum of Bernoulli random variables), 
Proposition~\ref{prop:one} and Proposition~\ref{prop:two} together give us Corollary~\ref{cor:samplesize}.
\end{section}

\begin{sidewaystable*}[]
\footnotesize
\begin{tabular}{l|llll|llll|llll}
\toprule
$\mathbf{\mathrm{H_2}}$ &  \multicolumn{4}{c}{Shots} & \multicolumn{4}{c}{Switches} & 
\multicolumn{4}{c}{Communications} \\
\midrule
Solver     & $Q_{25}$    & $Q_{50}$    & $Q_{75}$    & mean     & $Q_{25}$    & $Q_{50}$    & $Q_{75}$    & mean     & $Q_{25}$    & $Q_{50}$    & $Q_{75}$    & mean     \\ \hline
SHOALS     & $4.61\times 10^5$   & $2.31\times 10^6$  & $1.42\times 10^7$ & $9.09\times 10^6$  & 400      & \textbf{472}      & \textbf{815}      & \textbf{662}      & 23       & 28       & \textbf{45}       & \textbf{36}       \\
iCANS      & $\mathbf{2.74\times 10^5}$   & $6.83\times 10^5$   & $1.44\times 10^6$  & $1.55\times 10^6$  & 1200     & 1820     & 2480     & 1965     & 30       & 46       & 62       & 49       \\
Adam-100   & $5.68\times 10^5$   & $7.34\times 10^5$   & $\mathbf{9.44\times 10^5}$   & $\mathbf{7.36\times 10^5}$   & 5720     & 7380     & 9480     & 7403     & 143      & 185      & 237      & 185      \\
Adam-1000  & $5.64\times 10^6$  & $7.12\times10^6$  & $9.00\times 10^6$  & $7.13\times 10^6$  & 5680     & 7160     & 9040     & 7173     & 142      & 179      & 226      & 179      \\
LBFGS-100  & $\infty$ & $\infty$ & $\infty$ & $\infty$ & $\infty$ & $\infty$ & $\infty$ & $\infty$ & $\infty$ & $\infty$ & $\infty$ & $\infty$ \\
LBFGS-1000 & $3.15\times 10^5$   & $\mathbf{6.53\times 10^5}$   & $\infty$ & $\infty$ & \textbf{315}      & 652      & $\infty$ & $\infty$ & \textbf{7}        & \textbf{15}       & $\infty$ & $\infty$\\
\hline
$\mathbf{\mathrm{LiH}}$ &  \multicolumn{12}{c}{} \\
\hline
SHOALS     & $7.92\times 10^8$   & $1.27\times 10^9$  & $1.69\times 10^9$ & $1.41\times 10^9$  & $\mathbf{2.41\times 10^5}$      & $\mathbf{3.47\times 10^5}$      & $\mathbf{4.66\times 10^5}$      & $\mathbf{3.72\times 10^5}$      & \textbf{79}       & \textbf{110}       & \textbf{149}       & \textbf{122}      \\
iCANS      & $2.06\times 10^9$   & $2.75\times 10^9$   & $4.16\times 10^9$  & $2.86\times 10^9$  & $1.82\times 10^7$     & $2.04\times 10^7$     & $2.59\times 10^6$     & $2.15\times 10^6$     & 219       & 245       & 311       & 259       \\
Adam-100   & $\mathbf{1.61\times 10^8}$   & $\mathbf{2.25\times 10^8}$  & $\mathbf{3.29\times 10^8}$   & $\mathbf{2.36\times 10^8}$  & $1.61\times 10^6$     & $2.25\times 10^6$     & $3.29\times 10^6$     & $2.36\times 10^6$     & 193      & 271      & 396      & 283      \\
Adam-1000  & $1.11 \times 10^9$  & $1.77\times 10^9$  & $2.38\times 10^9$  & $1.83\times 10^9$  & $1.11\times 10^6$     & $1.77\times 10^6$     & $2.38\times 10^6$     & $1.83\times 10^6$     & 134     & 213      & 286      & 219      \\
\hline
$\mathbf{\mathrm{H_4}}$ &  \multicolumn{12}{c}{} \\
\hline
SHOALS     & $1.24\times 10^{11}$   & $1.36\times 10^{11}$  & $1.46\times 10^{11}$ & $1.35\times 10^{11}$  & $\mathbf{1.51\times 10^7}$      & $\mathbf{1.66\times 10^7}$      & $\mathbf{1.71\times 10^7}$     & $\mathbf{1.61\times 10^7}$      & \textbf{482}       & \textbf{525}       & \textbf{531}       & \textbf{508}      \\
iCANS      & $9.39\times 10^{11}$   & $1.04\times 10^{12}$   & $1.14\times 10^{12}$  & $1.04\times 10^{12}$  & $3.27\times 10^8$     & $3.43\times 10^8$     & $3.59\times 10^8$     & $3.43\times 10^8$     & 3086       & 3237       & 3388       & 3237       \\
Adam-100   & $\mathbf{1.29\times 10^{10}}$   & $\mathbf{1.39\times 10^{10}}$   & $\mathbf{1.46\times 10^{10}}$   & $\mathbf{1.38\times 10^{10}}$   & $1.29\times 10^8$    & $1.39\times 10^8$     & $1.46\times 10^8$     & $1.38\times 10^8$     & 1229      & 1309      & 1380      & 1301      \\
Adam-1000  & $1.09 \times 10^{11}$  & $1.32\times 10^{11}$  & $1.75\times 10^{11}$  & $1.41\times 10^{11}$  & $1.09\times 10^8$     & $1.32\times 10^8$     & $1.75\times 10^8$     & $1.41\times 10^8$     & 1032     & 1240      & 1649      & 1329      \\
\bottomrule
\end{tabular}
\caption{For each tested molecule and for each solver, we display the $25^{th}$, $50^{th}$, and $75^{th}$ percentiles ($Q_{25},Q_{50},Q_{75}$, respectively), as well as the mean, of the numbers of shots, switches, and communications needed to attain chemical accuracy in our experiments. Lowest values are highlighted in boldface font. We note that LBFGS-100 and LBFGS-1000 have undefined or ``infinite" values for each summary statistic for the two larger molecules and are hence excluded entirely. \label{table:summary}}
\end{sidewaystable*}

\framebox{\parbox{0.42\textwidth}{The submitted manuscript has been created by UChicago Argonne, LLC, Operator of Argonne National Laboratory (`Argonne'). Argonne, a U.S. Department of Energy Office of Science laboratory, is operated under Contract No. DE-AC02-06CH11357. The U.S. Government retains for itself, and others acting on its behalf, a paid-up nonexclusive, irrevocable worldwide license in said article to reproduce, prepare derivative works, distribute copies to the public, and perform publicly and display publicly, by or on behalf of the Government.  The Department of Energy will provide public access to these results of federally sponsored research in accordance with the DOE Public Access Plan. \url{http://energy.gov/downloads/doe-public-access-plan}.}}

\end{document}